\documentclass[english,hyphens]{article}
\usepackage[T1]{fontenc}
\usepackage[utf8x]{inputenc}
\usepackage{color}
\usepackage{float}
\usepackage{calc}
\usepackage{amsmath}
\usepackage{amsthm}
\usepackage{amssymb}
\usepackage{graphicx}

\makeatletter
\numberwithin{equation}{section}
\numberwithin{figure}{section}
\theoremstyle{plain}
\newtheorem{thm}{\protect\theoremname}
  \theoremstyle{remark}
  \newtheorem{rem}[thm]{\protect\remarkname}
  \theoremstyle{definition}
  \newtheorem{defn}[thm]{\protect\definitionname}
  \theoremstyle{plain}
  \newtheorem{cor}[thm]{\protect\corollaryname}

\usepackage{hyperref}
\hypersetup{colorlinks = false, allcolors=., citebordercolor = [rgb]{0,0,0}, linkbordercolor = [rgb]{0,0,0}, urlbordercolor = [rgb]{0,0,0}}
\usepackage{pifont}
\usepackage[T1]{fontenc}
\usepackage{microtype}
\usepackage{lipsum}
\usepackage{blindtext,titlefoot}
\usepackage{ragged2e}
\usepackage{float}
\usepackage{wasysym}
\usepackage{cjhebrew}
\usepackage{atbegshi,picture}

\usepackage{pgfplots}
\pgfplotsset{width=7cm,compat=1.10}

\newtheorem*{assumption*}{\assumptionnumber}
\providecommand{\assumptionnumber}{}
\makeatletter
\newenvironment{assumption}[1]
 {%
  \renewcommand{\assumptionnumber}{Assumption #1}%
  \begin{assumption*}%
  \protected@edef\@currentlabel{#1}%
 }
 {%
  \end{assumption*}
 }
\makeatother

\date{Yom Kippur\\ 10th of Tishrei, 5782\\ Shmitah and Yovel \footnote{According to Chatam Sofer's way of calculating the Yovel, 118*49=5782 (caveat: others may dispute the exact year, e.g., 5782+1).}}

\AtBeginShipout{\AtBeginShipoutUpperLeft{%
  \put(\dimexpr\paperwidth-1cm\relax,-1.5cm){\makebox[0pt][r]{\cjRL{h}"\cjRL{b}}}%
}}

\makeatother

\usepackage{babel}
  \providecommand{\corollaryname}{Corollary}
  \providecommand{\definitionname}{Definition}
  \providecommand{\remarkname}{Remark}
\providecommand{\theoremname}{Theorem}

\begin{document}

\title{{\Large{}JUBILEE: Secure Debt Relief and Forgiveness}}

\author{David Cerezo Sánchez\textsuperscript{}\\
{\small{}david@calctopia.com}}
\maketitle
\begin{abstract}
JUBILEE is a securely computed mechanism for debt relief and forgiveness
in a frictionless manner without involving trusted third parties,
leading to more harmonious debt settlements by incentivising the parties
to truthfully reveal their private information. JUBILEE improves over
all previous methods:

- individually rational, incentive-compatible, truthful/strategy-proof,
ex-post efficient, optimal mechanism for debt relief and forgiveness
with private information

- by the novel introduction of secure computation techniques to debt
relief, the ``\textit{blessing of the debtor}'' is hereby granted
for the first time: debt settlements with higher expected profits
and a higher probability of success than without using secure computation\\

A simple and practical implementation is included for ``The Secure
Spreadsheet''.

Another implementation is realised using Raziel smart contracts on
a blockchain with Pravuil consensus.\\


~

\end{abstract}

\section{Introduction}

As it is written in the Torah, the Jubilee is the remission year at
the end of seven cycles of seven years when slaves and prisoners are
freed, properties are returned to their original owners and debts
are forgiven \footnote{Deuteronomy 15:1–6, Leviticus 25:8–13, Isaiah 61:1-2}
to prevent profiting from the poor \footnote{Exodus 22:25}. The first
Jubilee was celebrated in AM 2553, starting the count 14 years after
entering Eretz  Yisrael \footnote{Rambam in Mishneh Torah, Sabbatical Year and the Jubilee 10:2},
as it was previously commandeth \footnote{Leviticus 25:2}: Jubilees
were announced with a \textit{shofar} \footnote{Ibn Ezra on Leviticus 25:9; Rashi on Leviticus 25:10},
an instrument made from a ram's horn (i.e., a \textit{yobhel} or \textit{yoveil,}
from which the term Jubilee is derived \footnote{Ramban on Leviticus 25:10})
and they provided a regulated and periodical ``clean state'' for
debt forgiveness.

By setting debt forgiveness on a fixed and periodical calendar, the
Jubilee resolved the age-old social problem of debt relief: modernly,
governments enforce an intricate system of debt collectors, courts
and legal procedures. In these cases, third parties are involved to
handle conflicts of interests between distrusting parties:
\begin{itemize}
\item debtors request too much debt forgiveness
\item creditors grant debt settlements as small as possible, trying to take
as many assets as possible from debtors: in some instances, the true
amount of debt relief needed is less than the expected revenue from
the debtor to fulfil their debt obligations, but the stated debt recovery
values from creditors exceed said expected revenue and/or continuation
value of the debtor, making it financially impossible for the debtor
to compensate the creditors. As a result, future projects from debtors
will not be undertaken, even if socially desirable, sacrificed to
satisfy the overstated demands from creditors.
\end{itemize}
For the first time, by combining mechanism design and secure computation
we manage to reconcile the demands of debtors and creditors in to
create a new debt relief mechanism in which the best course of action
is the truthful revelation of their private information: removing
most of the conflict from the situation, third parties are no longer
needed and frictionless debt settlements can be reached.

\subsection{Contributions}

In summary, we make the following contributions:
\begin{itemize}
\item an individually rational, incentive-compatible, truthful/strategy-proof,
ex-post efficient, optimal mechanism for debt relief and forgiveness
with private information is introduced
\item further, an implementation using secure computation is presented
\end{itemize}
In a nutshell, we contribute a new methodology for debt relief and
forgiveness in a frictionless and strifeless way without resorting
to third parties (i.e., courts, debt collectors), revamping an age-old
problem with modern cryptographic techniques such that truthful revelation
of private information is the best strategy for all parties involved.

In section \ref{sec:Related-Literature}, \textcolor{black}{we discuss
related literature and prior work. In section \ref{sec:Detailed-Design},
we first introduce our model, then proceed to describe a direct mechanism
for debt settlements, and finally describe the optimal mechanism.
In section \ref{sec:Practical-Implementation}, we describe our implementations
of the optimal mechanism and then we conclude in section \ref{sec:Conclusion}.}

The reader interested in the most practical applications may skip
to section \textcolor{black}{\ref{sec:Practical-Implementation},}

\section{\label{sec:Related-Literature}Related Literature}

The literature on the combination of secure computation and mechanism
design is surprisingly scarce: it mostly focuses on the problems of
secure computation with rational actors \cite{dcMeetsGT,rssIMD,rssMC,rssRevisited,cryptoeprint:2008:097},
and the combination of secure computation with the incentive-compatible
Vickrey auction or its generalisation as the Vickrey–Clarke–Groves
mechanism \cite{cryptoeprint:2003:021,Naor99privacypreserving,ppMultiUnitAuctions}.
Truthfulness is absolutely necessary for secure computation to gain
acceptance in the real world: although any truthful mechanism can
be securely computed \cite{cryptoeprint:2011:005}, the converse does
not hold. Thus, further research is needed to devise securely computed
mechanisms that encourage participants to report their information
truthfully.

\subsection{Prior Work in Cryptography}

Previous literature considered the use of secure computation technologies
for credit origination and rating \cite{byrd2020differentially,cryptoeprint:2015:1006,ppFinancialRisk,ppNetworkAnalytics,privacyPreservingCreditChecking},
but not for debt relief.

This paper also shows that the moral character of cryptographic work
\cite{cryptoeprint:2015:1162} goes way beyond preventing ``\textit{mass
surveillance}'', as it's currently customary in the field of cryptography.

\subsection{Prior Work in Game Theory}

Some seminal papers \cite{neillTalmud,aumannMaschler} started the
game-theoretic study of bankruptcy problems in the Talmud: see \cite{surveyMusketeers,surveyThomson}
for surveys of the extensive literature that followed, including alternative
derivations of the Mishnah of the Talmud rule \footnote{Rabbi Nathan in Babylonian Talmud, Kethuboth 93a; Yerushalmi Talmud,
Kethuboth 10,4} using cooperative bargaining \cite{bankruptcyCooperativeBargaining}
and the strategic Nash cooperative solution \cite{bankruptcyNashProgram}.
Note that the Talmud rule is so influential that it was adopted as
law (\textit{halakha}) in other contexts \footnote{Rambam in Mishneh Torah, Malveh veLoveh, Chapter 20, Section 4}
and it became a source of great discussion\footnote{Rabbi Seadia Gaon in Responsa Sha’arei Zedeq part 4, gate 4; Rabbi
Hai Gaon as quoted by Rabbi Isaac Alfasi on Kethuboth 93a; Rabbi Bezalel
Ashkenazi in Shitah Mekubetzet on Kethuboth 93a; Piniles H.M. in Darka
Shel Torah, p. 64; Rabbi Yehoshua Leib Diskin in Torat Ha’Ohel, vol.
1 page 2b} as Talmud scholars were trying to unravel and interpret it.

This paper fundamentally departs from this line of work found in the
game-theoretic literature:
\begin{enumerate}
\item it focuses on debt relief via debt settlements, it's not restricted
to equitable partitions between creditors
\item creditors are allowed to have private information (i.e., all information
is not public)
\item it uses cryptography as a primary tool to achieve its goals
\item it performs better than the Talmud rule in neutralising creditor's
conflict of interests
\end{enumerate}
The ultimate reason for this innovation is ``\textit{mi-p’nei tikkun
ha-olam}'' (i.e., for the bettering of the world, and for repairing
the world, as modernly interpreted): as participation is individually
rational for creditors (see next condition \ref{eq:creditorParticipationConstraint}),
it is economically rational to participate even with the institution
of the \textit{prozbul} \footnote{Gittin 4:3; Mishnah Sheviit 10},
a legal subterfuge to circumvent the commandeth forgiveness of loans
even after the Shmitah and the Yovel.

\section{\label{sec:Detailed-Design}Model and Design}

This paper is primarily concerned with a setting of distrusting parties,
a debtor entity and multiple creditors, holding asymmetric information
used as input in a process to settle non-marketable debt (i.e., creditors
are assumed not to trade between themselves). The debtor entity is
an intentional generalisation in order to abstract away the more specific
cases of a private party, a corporation or an indebted sovereign state:
note that each of these specific cases may need further refinements
to this general model. The results of this section are themselves
based on these works: \cite{myersonOptimalAuctionDesign}.

\subsection{Model Setup}

A highly indebted entity with total outstanding debt $D$ and no cash
in hand needs funds from external sources for an investment of $I$
(i.e., for debt service and to finance another project) that will
generate future cash flows given by the Discounted Cash Flow of the
current Asset continuation value plus the Investment value, $A+I$:
if the entity files for bankruptcy, the opportunity of the investment
in the other project will be lost. It is assumed that the entity is
prohibited by existing covenants from issuing new higher priority
debt senior to the existing outstanding debt: in other words, the
entity can only undertake the investment if it settles the outstanding
debt after successful negotiations with $n$ risk-neutral creditors
to fully or partially write down the principal of the debt since creditors
would be the main beneficiaries of the investment $I$, thus avoiding
bankruptcy and ensuring the continuation of the debtor entity.

To ease exposition, we assume that the risk-free interest rate is
zero and that each creditor has an identical share of the total outstanding
debt, $d=\left(D/n\right)$, each with equal priority. The willingness
to cancel some of the debt depends on two factors:
\begin{enumerate}
\item The estimated percentage of debt cancellation from other creditors.
\item The expected recovery value $\theta_{i}$ for each creditor $i$ in
case of bankruptcy if the entity does not receive any investments
and/or debt cancellation, and we assume that $\theta_{i}$ represents
the expected recovery value if the creditor owned claims to his proportional
part $n$ of the total entire value of the entity. 
\end{enumerate}
The creditor type $\theta_{i}$ is private information and everyone
holds identical expectations about the possible value of $\theta_{i}$,
captured by the random variable 
\[
\Theta_{i}:\left[\underline{\theta},\overline{\theta}\right]\rightarrow\left[0,1\right]
\]
with distribution $F\left(\theta_{i}\right)$, density $\phi\left(\theta_{i}\right)$,
and $\overline{\theta}\leq D$. With each creditor acting as Bayesian
decision-maker \cite{harsanyiBayesianGames}, creditor types are identically
and independently distributed, its list denoted by the vector $\theta=\left(\theta_{1},\ldots,\theta_{n}\right)$
and the set of all possible types of $s=1,2,\ldots,n$ creditors arise
from their different preferences, given by
\[
I_{s}=\left\{ \theta\mid\underline{\theta}\leq\theta_{j}\leq\overline{\theta}\,\text{for}\,j=1,2,\ldots,s\right\} 
\]

\begin{rem}
\label{rem:assumptions}We assume that each creditor $\theta_{i}$
are the only private knowledge, that is, the vector $\theta=\left(\theta_{1},\ldots,\theta_{n}\right)$
.

We further assume that $F\left(\theta_{i}\right)$, $D$, $A$, and
$I$ are common knowledge. There would be no uncertainty about how
much debt forgiveness should be granted if the vector $\theta$ was
publicly known.
\end{rem}

Additionally, we define
\begin{align*}
\phi\left(\theta\right) & =\prod_{j=1}^{n}\phi\left(\theta_{j}\right)\\
\theta_{-i} & =\left(\theta_{1},\ldots,\theta_{i-1},\theta_{i+1},\ldots,\theta_{n}\right)\\
\phi\left(\theta_{-i}\right) & =\prod_{j\neq i}\phi\left(\theta_{j}\right)
\end{align*}
Creditor $i$ may also know the expected recovery by other creditors
(i.e., using secure computation), in that case, the revised expected
recovery value in bankruptcy would be given by the following liquidation
value
\[
l_{i}\left(\theta_{i},\theta_{i-1}\right)=\theta_{i}+\sum_{j\neq i}e_{i}\left(\theta_{j}\right)
\]
with non-decreasing revised estimation function $e_{i}:\left[\underline{\theta},\overline{\theta}\right]\rightarrow\mathbb{R}$
defining how the creditor would revise his estimate if he knew the
estimate of creditor $j$ , and satisfying
\begin{equation}
\int_{\underline{\theta}}^{\overline{\theta}}e_{i}\left(\theta_{j}\right)\phi\left(\theta_{j}\right)d\theta_{j}=0,\,\text{such that }\forall j\neq i,\label{eq:revisionFunctionCondition}
\end{equation}
implying that $\theta_{i}$ would still be the expected recovery value
of $\theta_{i}$ because
\[
\int_{I_{n-1}}l_{i}\left(\theta_{i},\theta_{-i}\right)\phi\left(\theta_{-i}\right)d\theta_{-i}=\theta_{i}.
\]
As it is evident, the expected recovery value of $\theta_{i}$ cannot
exceed the full value of the debt,
\[
l_{i}\left(\theta_{i},\theta_{-i}\right)\leq d\,\,\text{such that }\forall\theta_{i}\in\left[\underline{\theta},\overline{\theta}\right],\theta_{-i}\in I_{n-1}.
\]

In this paper we use the following definition of efficiency:
\begin{defn}
(\textit{Ex-post} efficiency)\label{def:(ex-post-efficiency)}. We
consider a debt relief process by debt settlement to be \textit{ex-post}
efficient if and only if the entity remains solvent after the new
investment is undertaken with probability one when
\[
A\geq\sum_{i=1}^{n}l_{i}\left(\theta_{i},\theta_{-i}\right),
\]
and if the investment is undertaken with probability zero, the entity
goes bankrupt. Debt settlement can only be efficient if there is a
surplus in the difference between the value of an entity in solvency
and the value of the entity in bankruptcy,

Finally, note that the problem of debt settlement is a problem of
asymmetric information: if the debtor entity knew the true value of
$\theta_{i}$ for every creditor, it would be individually rational
to propose a settlement arrangement giving each creditor $i$ his
expected recovery value if and only if the settlement arrangement
is \textit{ex-post} efficient (i.e., continuation would be Pareto-efficient)
and all creditors would accept. Thus, asymmetric information is a
key feature of debt settlements and $\theta_{i}$ is assumed to be
private information: hence the use of secure computation, to enable
the privacy-preserving computation of said private information between
the distrusting parties.
\end{defn}

\subsection{A Direct Mechanism for Debt Settlement}

In this sub-section, we obtain an \textit{ex-post} efficient revelation
mechanism to settle the outstanding debt, by making use of the Revelation
Principle:
\begin{defn}
(Revelation Principle \cite{myersonOptimalAuctionDesign})\label{def:(Revelation-Principle)}.
Every equilibrium outcome of any arbitrary mechanism can be implemented
as an outcome in a truth-telling equilibrium of an incentive-compatible
direct revelation mechanism.

For each creditor, it's the best response to report his true type
in a truth-telling Bayesian equilibrium point of the revelation game,
thus creditors and the debtor entity would follow the instructions
of the following protocol:
\end{defn}

\noindent\fbox{\begin{minipage}[t]{1\columnwidth - 2\fboxsep - 2\fboxrule}%
\begin{center}
\textbf{Functionality $\mathcal{F}_{YOVEL}$}
\begin{figure}[H]
1. Each creditor privately reports their private information $\theta_{i}$
to a trusted party, at the same time.

2. Using the private reports, a trusted party calculates and instructs:
\begin{itemize}
\item the debtor entity to invest with probability $k$, or go bankrupt
with probability $1-k$. In case of continuing solvent, the entity
must make payments to the creditors as defined by the vector
\[
t=\left(t_{1},\ldots,t_{n}\right)
\]
\item each creditor $i$ must forgive an amount of debt equal to $d-t_{i}$
if the debtor is not declared bankrupt.
\end{itemize}
\caption{\label{fig:Ideal-functionality-Yovel}Ideal functionality $\mathcal{F}_{YOVEL}$}
\end{figure}
\par\end{center}%
\end{minipage}}

~\\
The vector of reported types is defined by
\[
\hat{\theta}=\left(\hat{\theta}_{1},\hat{\theta}_{2}\dots,\hat{\theta}_{n}\right),
\]
the recommended vector of transfer payments be denoted by $t:I_{n}\rightarrow\mathbb{R}^{n}$,
and the recommended probability $k:I_{n}\rightarrow\left[0,1\right]$
of continuing solvent and receiving investment.

Let $\Gamma=\left(t,k\right)$ denote the mechanism implemented by
the trusted party in the ideal model or a Secure Multi-Party Computation
program in the real model, then $t_{j}\left(\hat{\theta}_{i},\theta_{-i}\right)$
is the payment to creditor $j$, and $k\left(\hat{\theta}_{i},\theta_{-i}\right)$
is the investment probability when creditor $i$ reports $\hat{\theta}_{i}$
and all the creditors announce their true types. Note that for the
mechanism to be truthful equilibrium, it must be individually rational
for all players to participate and an equilibrium for each creditor
to report their true $\theta_{i}$.

The expected utility from a mechanism $\Gamma=\left(t,k\right)$ to
creditor $i$ such that it's an equilibrium to report their true type
for every other creditor, is given by
\[
\int_{I_{n-1}}\left\{ k\left(\hat{\theta}_{i},\theta_{-i}\right)t_{i}\left(\hat{\theta}_{i},\theta_{-i}\right)+\left[1-k\left(\hat{\theta}_{i},\theta_{-i}\right)\right]l_{i}\left(\theta_{i},\theta_{-i}\right)\right\} \phi\left(\theta_{-i}\right)d\theta_{-i},
\]
rewritten as
\[
\theta_{i}+\int_{I_{n-1}}k\left(\hat{\theta}_{i},\theta_{-i}\right)\left[t_{i}\left(\hat{\theta}_{i},\theta_{-i}\right)-l_{i}\left(\theta_{i},\theta_{-i}\right)\right]\phi\left(\theta_{-i}\right)d\theta_{-i}.
\]
using the definition of $l_{i}\left(\theta_{i},\theta_{-i}\right)$
and the properties of the revised estimation function \ref{eq:revisionFunctionCondition}.
The change in expected utility is given by
\begin{equation}
U\left(\theta_{i},\hat{\theta}_{i},t_{i},k\right)=\int_{I_{n-1}}k\left(\hat{\theta}_{i},\theta_{-i}\right)\left[t_{i}\left(\hat{\theta}_{i},\theta_{-i}\right)-l_{i}\left(\theta_{i},\theta_{-i}\right)\right]\phi\left(\theta_{-i}\right)d\theta_{-i}\label{eq:changeExpectedUtilityMechanism}
\end{equation}
because in case of no debt settlement, creditors get their expected
recovery value $\theta_{i}$ and the entity goes bankrupt; but in
case of staying solvent, creditors receive payments $t_{i}$ but they
lose the liquidation value. This change in expected utility must be
positive for all creditors

\begin{flalign}
\text{(IR)} &  &  & {U\left(\theta_{i},\hat{\theta}_{i},t_{i},k\right)\geq0,\,\text{such that }\forall\theta_{i}\in\left[\underline{\theta},\overline{\theta}\right],} & {}\label{eq:creditorParticipationConstraint}
\end{flalign}
for them to participate in the mechanism since debt relief can only
be granted willfully outside of bankruptcy, defining the individually
rational (IR) participation constraint: in other words, no party is
forced to participate in the mechanism, and each creditor is granted
veto power over the debt settlement via their joint control of the
probability $k\left(\right)$ thus effectively imposing negative externalities
on each other (note that extensions to this model introducing majority
voting rules are also possible).

On the other hand, the expected utility of the debtor entity from
the mechanism $\Gamma=\left(t,k\right)$ in a truth-telling equilibrium
is
\begin{equation}
V\left(\theta,k,t\right)=\int_{I_{n}}k\left(\theta\right)\left(A-\sum_{i=1}^{n}t_{i}\left(\theta\right)\right)\phi\left(\theta\right)d\theta,\label{eq:debtorExpectedUtility}
\end{equation}
also subject to a positive Ex-Ante Budget Balance (EXABB) participation
constraint

\begin{flalign}
\text{(EXABB)} &  &  & {V\left(\theta,k,t\right)\geq0} & {}\label{eq:debitorParticipationConstraint}
\end{flalign}

An additional incentive-compatibility constraint (IC) for the mechanism
$\Gamma=\left(t,k\right)$ must be satisfied for truthful reporting
to be an equilibrium, which in conjunction with the individually rational
(IR) participation constraint \ref{eq:creditorParticipationConstraint}
define the mechanism as a feasible mechanism: 

\begin{flalign}
\text{(IC)} &  &  & {U\left(\theta_{i},\hat{\theta}_{i},t_{i},k\right)\geq U\left(\theta_{i},\theta_{j},t_{i},k\right),\,\text{such that }\forall\theta_{i},\theta_{j}\in\left[\underline{\theta},\overline{\theta}\right],\forall i,j.} & {}\label{eq:incentive-compatible-truthful-constraint}
\end{flalign}

To obtain incentive-compatibility \cite{myersonOptimalAuctionDesign},
it's a first necessary and sufficient condition that
\[
K\left(\theta_{i}\right)=\int_{I_{n-1}}k\left(\theta_{i},\theta_{i-1}\right)\phi\left(\theta_{-i}\right)d\theta_{-i}
\]
must be decreasing in $\theta_{i}$ for all $\theta_{i}\in\left[\underline{\theta},\overline{\theta}\right]$.
In other words, the higher the expected recovery value of a creditor
the smaller the probability that the entity will remain solvent as
expected by that same creditor. The other second necessary and sufficient
condition derived from \cite{myersonOptimalAuctionDesign} is that
$\forall\theta_{i}$,
\begin{equation}
U\left(\theta_{i},\theta_{i},t_{i},k\right)=U\left(\bar{\theta},\bar{\theta},t_{i},k\right)+\int_{I_{n-1}}\int_{\theta_{i}}^{\bar{\theta}}k\left(u,\theta_{-i}\right)du\phi\left(\theta_{-i}\right)d\theta_{-i}\geq0,\label{eq:secondConditionIncentiveCompatibility}
\end{equation}
granting, for the highest type $\overline{\theta}$, the expected
change in utility to each creditor $i$ plus a positive mark-up: the
new rightmost term of the right-hand side is an economic incentive
given to creditors of type $\theta_{i}<\bar{\theta}$ by the mechanism
to induce them to reveal that they are the creditors more inclined
to grant debt forgiveness (i.e., an ``informational rent''). To
derive the equation \ref{eq:secondConditionIncentiveCompatibility},
we start denoting by $\hat{U}\left(\theta_{i,}\theta_{i,},k,t\right)$
as the maximum utility of creditor $i$ from the mechanism $\left(k,t\right)$:
then, by the envelope theorem\cite{samuelsonFundamentals,milgromAuctionTheory,generalisedEnvelope},
\[
\frac{d\hat{U}}{d\theta_{i}}=-\int_{I_{n-1}}k\left(\theta_{i},\theta_{-i}\right)\phi\left(\theta_{-i}\right)d\theta_{-i}\leq0
\]
Therefore, by reintegration we obtain the previous equation \ref{eq:secondConditionIncentiveCompatibility}
\begin{align*}
\hat{U}\left(\theta_{i},\theta_{i},k,t\right) & =U\left(\bar{\theta},\bar{\theta},t,k\right)+\int_{I_{n-1}}\left(\int_{\theta_{i}}^{\bar{\theta}}\frac{d\hat{U}}{du}du\right)\phi\left(\theta_{-i}\right)d\theta_{-i}\\
 & =U\left(\bar{\theta},\bar{\theta},t,k\right)+\int_{I_{n-1}}\int_{\theta_{i}}^{\bar{\theta}}p\left(u,\theta_{-i}\right)du\phi\left(\theta_{-i}\right)d\theta_{-i}
\end{align*}

Now let's look at the incentive-compatible mechanism most favourable
to the debtor, that is, when $U\left(\bar{\theta},\bar{\theta},t_{i},k\right)=0$:
then, the amount the debtor will pay is the expected recovery value
plus an additional economic incentive derived from \ref{eq:secondConditionIncentiveCompatibility},
\begin{equation}
\int_{I_{n}}k\left(\theta\right)t_{i}\left(\theta_{i},\theta_{-i}\right)\phi\left(\theta\right)d\theta=\int_{I_{n}}k\left(\theta\right)\left[l_{i}\left(\theta_{i},\theta_{-i}\right)+\frac{F\left(\theta_{i}\right)}{\phi\left(\theta_{i}\right)}\right]\phi\left(\theta\right)d\theta\label{eq:debtorExpectedProfits}
\end{equation}

To derive equation \ref{eq:debtorExpectedProfits}, we start from
\ref{eq:changeExpectedUtilityMechanism} and \ref{eq:secondConditionIncentiveCompatibility},
\begin{align*}
\int_{I_{n-1}}k\left(\theta_{i},\theta_{-i}\right)t_{i}\left(\theta_{i},\theta_{-i}\right)\phi\left(\theta_{-i}\right)d\theta_{-i} & =\int_{I_{n-1}}k\left(\theta_{i},\theta_{-i}\right)l_{i}\left(\theta_{i},\theta_{-i}\right)\phi\left(\theta_{-i}\right)d\theta_{i}\\
 & +\int_{I_{n-1}}\int_{\theta_{i}}^{\bar{\theta}}k\left(u,\theta_{-i}\right)du\phi\left(\theta_{-i}\right)d\theta_{-i}
\end{align*}
By taking expectations over all possible values of $\theta_{i}$,
we obtain
\begin{align}
\int_{I_{n}}k\left(\theta\right)t_{i}\left(\theta_{i},\theta_{-i}\right)\phi\left(\theta\right)d\theta & =\int_{I_{n}}k\left(\theta_{i},\theta_{-i}\right)l_{i}\left(\theta_{i},\theta_{-i}\right)\phi\left(x\right)d\theta\nonumber \\
 & +\int_{I_{n.1}}\left[\int_{\underline{\theta}}^{\overline{\theta}}\int_{\theta_{i}}^{\overline{\theta}}k\left(u,\theta_{-i}\right)du\phi\left(\theta_{i}\right)d\theta_{i}\right]\phi\left(\theta_{i}\right)d\theta_{-i}\label{eq:derivationExpectedInformationalRent}
\end{align}
Integrating by parts the term in brackets, we arrive at
\[
\int_{\underline{\theta}}^{\overline{\theta}}\int_{\theta_{i}}^{\overline{\theta}}k\left(u,\theta_{-i}\right)du\phi\left(\theta_{i}\right)d\theta_{i}=\int_{\underline{\theta}}^{\overline{\theta}}k\left(\theta_{i},\theta_{-i}\right)F\left(\theta_{i}\right)d\theta_{i}
\]
Thus, equation \ref{eq:derivationExpectedInformationalRent} can be
rewritten as the desired equation \ref{eq:debtorExpectedProfits}

\begin{align*}
\int_{I_{n}}k\left(\theta\right)t_{i}\left(\theta_{i},\theta_{-i}\right)f\left(\theta\right)d\theta & =\int_{I_{n}}k\left(\theta_{i},\theta_{-i}\right)l_{i}\left(\theta_{i},\theta_{-i}\right)f\left(\theta\right)d\theta\\
 & +\int_{I_{n}}k\left(\theta_{i},\theta_{-i}\right)\left[F\left(\theta_{i}\right)/\phi\left(\theta_{i}\right)\right]\phi\left(\theta\right)d\theta\\
 & =\int_{I_{n}}k\left(\theta\right)\left[l_{i}\left(\theta_{i},\theta_{-i}\right)+\frac{F\left(\theta_{i}\right)}{\phi\left(\theta_{i}\right)}\right]\phi\left(\theta\right)d\theta
\end{align*}

By substituting the previous relation on the debtor's expected utility
\ref{eq:debtorExpectedUtility}, it becomes 
\begin{equation}
\int_{I_{n}}k\left(\theta\right)\left(A-\sum_{i=1}^{n}\left[l_{i}\left(\theta_{i},\theta_{-i}\right)+\frac{F\left(\theta_{i}\right)}{\phi\left(\theta_{i}\right)}\right]\right)\phi\left(\theta\right)d\theta\label{eq:newDebtorsExpectedUtility}
\end{equation}

\begin{thm}
\label{thm:positiveDebtorExpectedUtility} A mechanism $\Gamma=\left(t,k\right)$
is incentive-compatible according to constraint \ref{eq:incentive-compatible-truthful-constraint}
(IC) and satisfies the creditor's participation constraint \ref{eq:creditorParticipationConstraint}
(IR) and the debtor's participation constraint \ref{eq:debitorParticipationConstraint}
(EXABB) if and only if the debtor's expected utility is strictly positive,
\begin{equation}
\int_{I_{n}}k\left(\theta\right)\left(A-\sum_{i=1}^{n}\left[l_{i}\left(\theta_{i},\theta_{-i}\right)+\frac{F\left(\theta_{i}\right)}{\phi\left(\theta_{i}\right)}\right]\right)\phi\left(\theta\right)d\theta\geq0\label{eq:positiveDebtorExpectedUtility}
\end{equation}
\end{thm}

\begin{proof}
Note that by the debtor's participation constraint \ref{eq:debitorParticipationConstraint}
(EXABB),
\[
V\left(\theta,k,t\right)\geq0
\]
and that \ref{eq:newDebtorsExpectedUtility} is derived from the debtor's
expected utility \ref{eq:debtorExpectedUtility}, then it trivially
follows that \ref{eq:positiveDebtorExpectedUtility} must also be
positive.
\end{proof}
\begin{defn}
(\textit{Blessing of the debtor})\label{def:(Debtor's-blessing).}.
The willingness of creditors to grant debt forgiveness is higher when
using secure computation for their private information. This paradoxical
situation is the opposite of the ``\textit{winner's curse}'' from
auction theory \cite{counteringWinnerCurse}: given that an agreement
on debt forgiveness among creditors must be unanimous \ref{eq:creditorParticipationConstraint}
, then debt forgiveness could only happen if all creditors are willing
to grant forgiveness because their expected recovery value from bankruptcy
is low, thus the willingness of each creditor gets reinforced from
the private information obtained from other creditors.
\end{defn}

\begin{thm}
\label{thm:privateInformationHigherDebtorExpectedProfits} The expected
profits of the debtor are higher when there are differences in private
information, and not just differences in preferences, for any incentive-compatible
$k\left(\theta\right)$.
\end{thm}

\begin{proof}
A creditor is more willing to grant debt forgiveness when there are
differences in private information (i.e., \textit{the blessing of
the debtor} \ref{def:(Debtor's-blessing).}) because the expected
recovery value conditional on a successful settlement is smaller than
the unconditional expected recovery value, given that all creditors
consider that debt settlements will only prosper when other debtholders
expect a low recovery value. This allows the debtor entity to extract
more debt forgiveness from creditors under private information.

Debtor's expected profits are given by \ref{eq:debtorExpectedProfits}
if the investment rule $k\left(\theta\right)$ is incentive-compatible.
When there are only differences in preferences, the revision functions
are $e_{i}\left(\theta_{j}\right)=0$ for all values of $i,j,$ and
$\theta_{j}$, thus we only need to show that 
\[
\int_{I_{n}}k\left(\theta\right)\left[\sum_{i=1}^{n}e_{i}\left(\theta_{-i}\right)\right]\phi\left(\theta\right)d\theta<0
\]
with $e_{i}\left(\theta_{-i}\right)=\sum_{j\neq i}e_{i}\left(\theta_{j}\right)$.
By the definition of $K\left(\theta_{i}\right)$, the inequality can
be rewritten as
\[
\sum_{i=1}^{n}\sum_{j\neq i}\left[\int_{\underline{\theta}}^{\bar{\theta}}K\left(\theta_{j}\right)e_{i}\left(\theta_{j}\right)\phi\left(\theta_{j}\right)d\theta_{j}\right]<0.
\]

Since $e_{i}\left(\theta_{i}\right)$ is increasing and by definition
$\int_{\underline{\theta}}^{\overline{\theta}}e_{i}\left(\theta_{j}\right)\phi\left(\theta_{j}\right)d\theta_{j}=0$
for all $i$ and $j$, there exists an $\hat{\theta}_{i}$ such that
$e_{i}\left(\theta_{j}\right)\lessgtr0$ as $\theta_{j}\lessgtr\hat{\theta}_{i}$,
and we can write
\[
\int_{\underline{\theta}}^{\hat{\theta_{i}}}K\left(\hat{\theta}_{i}\right)e_{i}\left(\theta_{j}\right)\phi\left(\theta_{j}\right)d\theta_{j}+\int_{\hat{\theta}_{i}}^{\bar{\theta}}K\left(\hat{\theta}_{i}\right)e_{i}\left(\theta_{j}\right)\phi\left(\theta_{j}\right)d\theta_{j}=0
\]

Due to incentive-compatibility we know that $K\left(\theta_{j}\right)$
is decreasing, hence 
\[
\int_{\underline{\theta}}^{\hat{\theta}_{i}}K\left(\theta_{i}\right)e_{i}\left(\theta_{j}\right)\phi\left(\theta_{j}\right)d\theta_{j}<\int_{\underline{\theta}}^{\hat{\theta}_{i}}K\left(\hat{\theta}_{i}\right)e_{i}\left(\theta_{j}\right)\phi\left(\theta_{j}\right)d\theta_{j},
\]
and
\[
\int_{\hat{\theta_{i}}}^{\overline{\theta}}K\left(\hat{\theta}_{i}\right)e_{i}\left(\theta_{j}\right)\phi\left(\theta_{j}\right)d\theta_{j}>\int_{\hat{\theta_{i}}}^{\overline{\theta}}K\left(\theta_{i}\right)e_{i}\left(\theta_{j}\right)\phi\left(\theta_{j}\right)d\theta_{j}.
\]
But then 
\[
\int_{\underline{\theta}}^{\bar{\theta}}K\left(\theta_{j}\right)e_{i}\left(\theta_{j}\right)\phi\left(\theta_{j}\right)d\theta_{j}<K\left(\hat{\theta}_{i}\right)\int_{\underline{\theta}}^{\bar{\theta}}e_{i}\left(\theta_{j}\right)\phi\left(\theta_{j}\right)d\theta_{j}=0.
\]
\end{proof}
A recent impossibility result \cite{jehielMoldavanu} further limits
mechanisms' ability to implement efficient allocations when creditors'
private values aren't private information (i.e., with secure computation
as used here).
\begin{thm}
(Jehiel-Moldovanu Impossibility Theorem \cite{jehielMoldavanu,milgromAuctionTheory}).\label{thm:Jehiel-Moldovanu-Impossibility}
Let $\Gamma(\theta)$ be a mechanism where the liquidation value $l_{i}\left(\theta_{i},\theta_{i-1}\right)$
depends on the valuations of other creditors but without private information,
and suppose that the function $E\left(\theta^{i}\right)\equiv E\left[\Gamma^{i}\left(\theta\right)|\theta^{i}\right]$
depends non-trivially on $\theta_{-i}^{i}$. Then, no mechanism exists
that implements $\Gamma$ at any Bayes-Nash equilibrium.
\end{thm}

\begin{cor}
The use of secure computation techniques enables debt settlements
with higher expected profits, thus attaining the ``blessing of the
debtor'' \ref{def:(Debtor's-blessing).} with efficient allocations
for creditors.
\end{cor}

\begin{proof}
Follows from previous Theorem \ref{thm:privateInformationHigherDebtorExpectedProfits},
since secure computation enables computation on private information
between multiple distrusting parties. Additionally, efficient allocations
for creditors are only possible with private information by Theorem
\ref{thm:Jehiel-Moldovanu-Impossibility}.
\end{proof}
Additionally Theorem \ref{thm:privateInformationHigherDebtorExpectedProfits}
implies that condition \ref{eq:positiveDebtorExpectedUtility} of
Theorem \ref{thm:positiveDebtorExpectedUtility} will be satisfied
with higher probability when there are differences in private information,
making debt settlements more efficient.

\subsection{An Optimal Revelation Mechanism}

We start choosing the investment rule $k\left(\theta\right)$ to maximise
the debtor's expected utility given by \ref{eq:newDebtorsExpectedUtility},
\[
\hat{k}\left(\theta\right)=\begin{cases}
1 & \text{if }A\geq\sum_{i=1}^{n}\left[l_{i}\left(\theta_{i},\theta_{-i}\right)+\frac{F\left(\theta_{i}\right)}{\phi\left(\theta_{i}\right)}\right]\\
0 & \text{otherwise}
\end{cases}
\]
In other words, the entity will remain solvent if the continuation
value is higher than the sum of the expected recovery values in bankruptcy,
plus the terms $\left(F\left(x_{i}\right)/\phi\left(\theta_{i}\right)\right)$:
this investment rule $\hat{k}\left(\theta\right)$ and the condition
that $U\left(\bar{\theta},\bar{\theta},t,k\right)=0$ would set debtor's
expected profits. But we also need an explicit solution for $t\left(\theta\right)$,
thus the investment rule $\hat{k}\left(\theta\right)$ is rewritten
as
\[
\hat{k}\left(\theta\right)=\begin{cases}
1 & \text{if }C\geq B\left(\theta_{i}\right)+Q\left(\theta_{-i}\right)\\
0 & \text{otherwise}
\end{cases}
\]
where
\[
B\left(\theta_{i}\right)=\theta_{i}+\frac{F\left(\theta_{i}\right)}{\phi\left(\theta_{i}\right)}+\sum_{j\neq i}e_{j}\left(\theta_{i}\right)
\]
and
\[
Q\left(\theta_{-i}\right)=\sum_{j=1}^{n}\sum_{k\neq j,i}e_{j}\left(\theta_{k}\right)+\sum_{j\neq i}\left[\theta_{j}+\frac{F\left(\theta_{j}\right)}{\phi\left(\theta_{j}\right)}\right]
\]
We make the following assumption so $B\left(x_{i}\right)$ is strictly
increasing.

\begin{assumption}{1} $F\left(\theta_{i}\right)/\phi\left(\theta_{i}\right)$ is strictly increasing (i.e., monotonically non-decreasing) for all $\theta_{i}\in\left[\underline{\theta},\overline{\theta}\right]$ and all $i=1,\dots,N$.    \end{assumption}

This assumption is satisfied by any uniform distribution on the interval
$0\leq\underline{\theta}<\overline{\theta}<\infty$, the Pareto, exponential,
and positive normal distributions.

Thus, there exists a unique value of $\theta_{i}$ denoted by $\tilde{\theta}=\tilde{\theta}\left(\theta_{-i}\right)$
for each vector $\theta_{-i}$ that solves 
\[
A=B\left(\theta_{i}\right)+Q\left(\theta_{-i}\right)
\]
The pivotal type $\tilde{\theta}\left(\theta_{-i}\right)$ that creditor
$i$ can report without forcing bankruptcy according to investment
rule $\hat{k}\left(\theta\right)$ is large when low types $\theta_{-i}$
are being reported by the other creditors: it measures the magnitude
of the opportunities for holding out. It's defined by
\[
\tilde{\theta}\left(\theta_{-i}\right)=\left\{ \text{min }\theta_{i}|A=\sum_{i=1}^{n}\left[\theta_{i}+e\left(\theta_{-i}\right)+\frac{F\left(\theta_{i}\right)}{\phi\left(\theta_{i}\right)}\right]\right\} 
\]

\begin{thm}
\label{thm:optimal-mechanism}The payment to creditor $i$ is calculated
according to
\[
\hat{t}_{i}\left(\theta_{i},\theta_{-i}\right)=\hat{t}_{i}\left(\theta_{-i}\right)=\tilde{\theta}\left(\theta_{-i}\right)+e_{i}\left(\theta_{-i}\right)
\]
in the optimal mechanism.
\end{thm}

\begin{proof}
The investment rule $\hat{k}\left(\theta\right)$ can be rewritten
by the definition of $\tilde{\theta}\left(\theta_{-i}\right)$ as
\[
k\left(\theta\right)=\begin{cases}
1 & \text{for }\theta_{i}\leq\tilde{\theta}\left(\theta_{-i}\right)\\
0 & \text{otherwise}
\end{cases}
\]
Thus, the change in utility to creditor $i$ from the optimal mechanism
by constraint \ref{eq:incentive-compatible-truthful-constraint} is
\begin{align}
U\left(\theta_{i},\theta_{i},\hat{k},\hat{t}\right) & =\int_{I_{n-1}}\int_{\theta_{i}}^{\bar{\theta}}\hat{k}\left(u,\theta_{-i}\right)du\phi\left(\theta_{-i}\right)d\theta_{-i}\\
 & =\int_{I_{n-1}}\text{max}\left[\tilde{\theta}\left(\theta_{-i}\right)-\theta_{i},0\right]\phi\left(\theta_{-i}\right)d\theta_{-i}\label{eq:derivationStep1}
\end{align}

From the previous equation \ref{eq:derivationStep1} and the change
in the expected utility from the non-optimal mechanism \ref{eq:changeExpectedUtilityMechanism},
\begin{align}
 & \int_{I_{n-1}}\hat{k}\left(\theta_{i},\theta_{-i}\right)t_{i}\left(\theta_{i},\theta_{-i}\right)\phi\left(\theta_{-i}\right)d\theta_{-i}\\
 & =\int_{I_{n-1}}\left(\theta_{i}+e_{i}\left(\theta_{-i}\right)+\text{max}\left[\tilde{\theta}\left(\theta_{-i}\right)-\theta_{i},0\right]\right)\phi\left(\theta_{-i}\right)d\theta_{-i}\label{eq:derivationStep2}
\end{align}

Finally, we arrive to the solution to equation \ref{eq:derivationStep2},
\[
\hat{t}_{i}\left(\theta_{i},\theta_{-i}\right)=\hat{t}_{i}\left(\theta_{-i}\right)=\tilde{\theta}\left(\theta_{-i}\right)+e_{i}\left(\theta_{-i}\right)
\]
since $\tilde{\theta}\left(\theta_{-i}\right)-\theta_{i}<0$ when
$\hat{k}\left(\theta_{i},\theta_{-i}\right)=0$, and $\tilde{\theta}\left(\theta_{-i}\right)-\theta_{i}>0$
when $\hat{k}\left(\theta_{i},\theta_{-i}\right)=1$.
\end{proof}
In the optimal mechanism \ref{thm:optimal-mechanism}, a payment is
received by creditor $i$ which is independent of the reported type,
but it will depend on the types declared by other creditors as explained
in the following:
\begin{itemize}
\item the function $\tilde{\theta}\left(\theta_{-i}\right)$ is decreasing
in each component of $\theta_{-i}$ .
\item the function $e_{i}\left(\theta_{-i}\right)$ is increasing in each
component of $\theta_{-i}$ : if other creditors report high expected
recovery values, then creditor $i$ would increase his own expected
recovery value.
\end{itemize}
If the holding out from $\tilde{\theta}\left(\theta_{-i}\right)$
dominates the revision of expectations from $e_{i}\left(\theta_{-i}\right)$,
then transfers are decreasing in each component of $\theta_{-i}$
when using the optimal mechanism.

\section{\label{sec:Practical-Implementation}Practical Implementation}

To ease exposition, we consider the simplest case: consider an example
with $n=2$ creditors and their private expected recovery value from
bankruptcy $\theta_{i}$ be uniformly distributed on $\left[0,1\right]$. 

Assume that the revision function \ref{eq:revisionFunctionCondition}
is defined by
\[
e_{i}\left(\theta_{i}\right)=\alpha\left(\theta_{i}-\frac{1}{2}\right)
\]
 such that when the constant $\alpha>1$ creditors give more weight
to the estimate of other creditors, and vice versa. 

The optimal continuation value is given by
\begin{equation}
\hat{k}\left(\theta\right)=\begin{cases}
1 & \text{if }A\geq\left(\theta_{1}+\theta_{2}\right)\left(1+\alpha\right)-\alpha\\
0 & \text{otherwise}
\end{cases}\label{eq:derivedOptimalContinuationValue}
\end{equation}
\textcolor{black}{thus}
\[
\tilde{\theta}\left(\theta_{2}\right)=\left(\frac{4}{2+\alpha}\right)\left(A-\frac{\alpha}{2}\right)-2\theta_{1}
\]
the payment to creditor 1 is
\begin{equation}
\hat{t}_{1}\left(\theta_{2}\right)=\tilde{\theta}\left(\theta_{2}\right)+e\left(\theta_{2}\right)=\left(\frac{4A-\alpha^{2}}{2\alpha+4}\right)-\theta_{2}\left(1-\alpha\right)\label{eq:derivedPaymentCreditor1}
\end{equation}
and the payment to creditor 2 is
\begin{equation}
\hat{t}_{2}\left(\theta_{1}\right)=\tilde{\theta}\left(\theta_{1}\right)+e\left(\theta_{1}\right)=\left(\frac{4A-\alpha^{2}}{2\alpha+4}\right)-\theta_{1}\left(1-\alpha\right)\label{eq:derivedPaymentCreditor2}
\end{equation}
Finally, the amount of debt forgiveness by each creditor $i$ is given
by $d-\theta_{i}$.

\subsection{Implementation on ``The Secure Spreadsheet''}

\begin{figure}[H]
\begin{centering}
\label{screenshots}
\par\end{centering}
\begin{centering}
\includegraphics[scale=0.5]{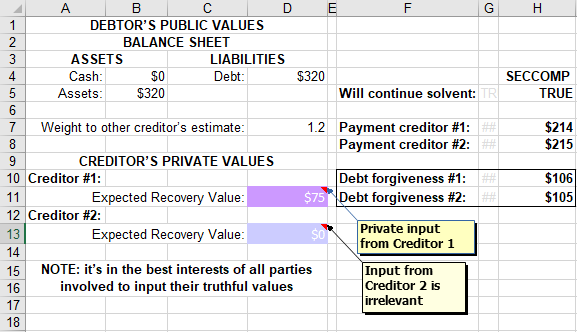}
\par\end{centering}
\begin{centering}
\includegraphics[scale=0.5]{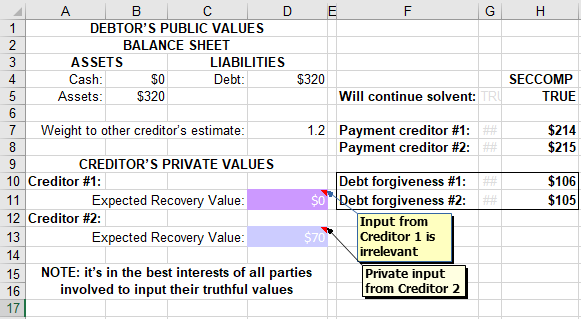}
\par\end{centering}
\caption{Screenshots from ``The Secure Spreadsheet''. Top is creditor \#1,
bottom is creditor \#2.}

\end{figure}

The previously derived closed-form formulae for the probability of
the debtor to receive debt relief and continue being solvent \ref{eq:derivedOptimalContinuationValue},
the payment to creditor \#1 \ref{eq:derivedPaymentCreditor1} and
to creditor \#2 \ref{eq:derivedPaymentCreditor2} can be securely
computed very easily on ``The Secure Spreadsheet'' (USPTO Patent
10,423,806\cite{secureSpreadsheetPatent}), available online for download
\cite{secureSpreadsheetProgram}: first precognised in the article
``The G-d Protocols'' \cite{gdProtocols} from 1997, ``The Secure
Spreadsheet'' is the first and only user program for general-purpose
secure computation. For maximum performance at an affordable cost,
the preferred secure computation protocol used by the ``The Secure
Spreadsheet'' is the \textit{dual-execution} protocol\cite{tradeoffs2party}:
\begin{thm}
(Dual-Execution protocol\cite{quidProQuo}). If the garbled circuit
construction is secure against semi-honest adversaries and the hash
function $\text{M}$ is modelled as a random oracle, then the Dual-Execution
protocol securely computes $f$ implementing the ideal functionality
$\mathcal{F}_{YOVEL}$ \ref{fig:Ideal-functionality-Yovel} if for
every non-uniform probabilistic polynomial-time adversary $\mathcal{A}$
in the real model, there exists a non-uniform probabilistic polynomial-time
adversary $\mathcal{S}$ in the ideal model such that
\[
\left\{ \text{IDEAL}_{f,S\left(\text{aux}\right)}\left(x,y,n\right)_{x,y,aux\in\left\{ 0,1\right\} ^{*}}\right\} \overset{c}{\equiv}\left\{ \text{REAL}_{\prod,\mathcal{A}\left(\text{aux}\right)}\left(x,y,n\right)\right\} _{x,y,\text{aux}\in\left\{ 0,1\right\} ^{n}}
\]
\end{thm}

As pictured above on the captured screenshots \ref{screenshots},
the public values for the debt $D$, the continuation value $C$ and
the weight $\alpha$ to the estimate to the other creditors must be
the same on both spreadsheets for creditors 1 and 2. On the top spreadsheet
for creditor 1, the only private that is taken into account is the
one inputted by creditor 1 (et vice versa for the bottom spreadsheet
for creditor 2). The final values computed using secure computation
appear on the right-hand side of the captured screenshots, under the
column ``SECCOMP''.

\subsection{Implementation on a blockchain}

The previous implementation on spreadsheets \ref{screenshots} is
also realised on a blockchain using Raziel \cite{Raziel} smart contracts
and Pravuil \cite{Pravuil} consensus: identical secure computations
as the ones carried out on ``The Secure Spreadsheet'' are executed
among creditors and debtors interfacing the blockchain with a mobile
app.

\begin{figure}[H]
\begin{centering}
\includegraphics[scale=0.4]{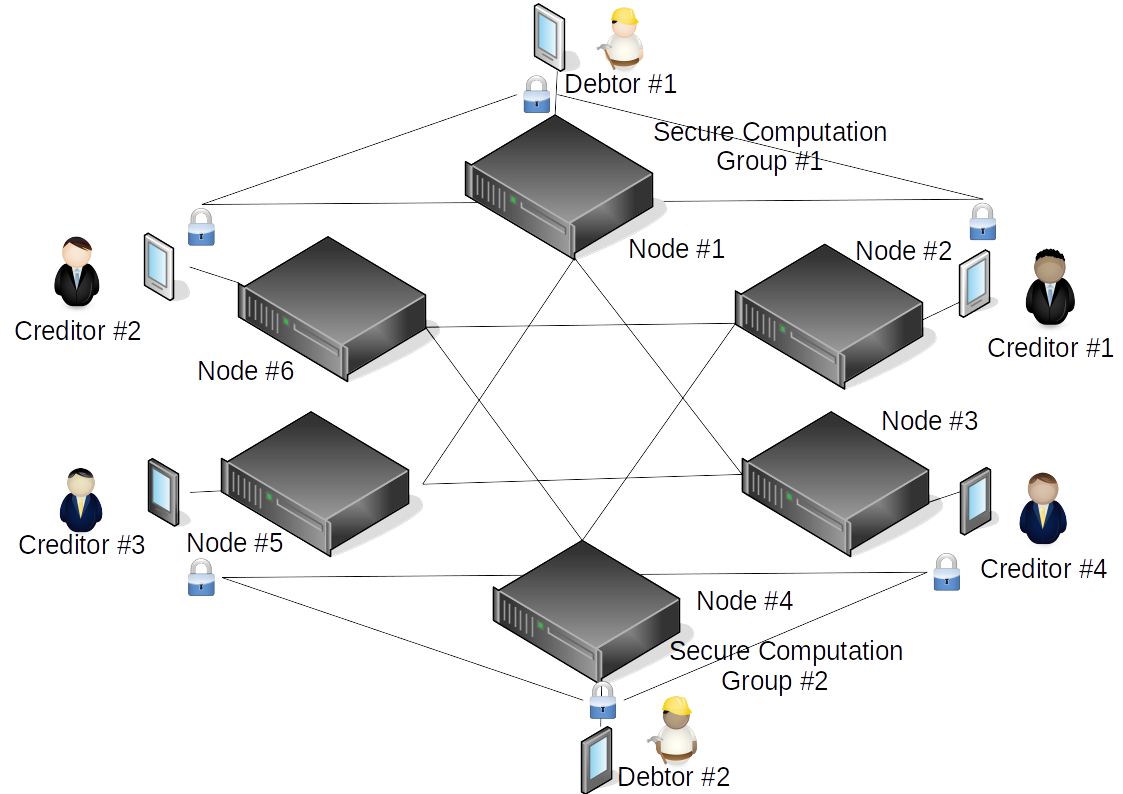}
\par\end{centering}
\caption{Two secure computation groups forgiving debts on a blockchain}

\end{figure}

Note the special suitability of Pravuil \cite{Pravuil} consensus
integrating real-world identity on this blockchain: debt management
requires real-world identities, otherwise Sybil attacks could create
unlimited fake debts. Due to this reason, this is the only valid blockchain
consensus protocol \cite{Pravuil} for the purpose of debt management.

\section{\label{sec:Conclusion}Conclusion}

The present paper has tackled and successfully solved the problem
of debt relief and forgiveness by securely computing optimal debt
settlements. The mathematical proofs hereby provided demonstrate that
participation in the proposed mechanism is within the economically
rational interests of all the involved parties by truthfully providing
their private information, thus removing the need for third parties.
Additionally, the provided implementations demonstrate the practicality
of the securely computed mechanism.

\bibliographystyle{alpha}
\bibliography{bib}

\end{document}